\documentclass[12pt]{article}
\makeatletter
\newcommand{\singlespacing}{\let\CS=\@currsize\renewcommand{\baselinestretch}{1}\tiny\CS}
\usepackage{epsfig}
\usepackage{amsthm}
\usepackage{graphicx}
\usepackage{verbatim}   
\usepackage{color}      
\usepackage{amsfonts}
\usepackage{amsmath}
\usepackage{makecell}

\usepackage{caption}
\newtheorem{theorem}{Theorem}

\newtheorem{corollary}{Corollary}
\begin{document}
\baselineskip=24pt
\parskip = 10pt
\def \qed {\hfill \vrule height7pt width 5pt depth 0pt}

\title{\sc On Generalized Progressive Hybrid Censoring in presence of competing risks}
\author{\sc Arnab Koley\footnote{Department of Mathematics and Statistics, Indian Institute of
Technology Kanpur, Pin 208016, India.} \& Debasis Kundu$^*$\footnote{Corresponding author, E-mail: kundu@iitk.ac.in, Phone no. 91-512-2597141, Fax no. 91-512-2597500.}}

\date{}
\maketitle

\begin{abstract}
\noindent  The progressive Type-II hybrid censoring scheme introduced by Kundu and Joarder (\textit{Computational Statistics 
and Data Analysis}, 2509-2528, 2006), has received some attention in the last few years.  One major drawback of this 
censoring scheme is that very few observations (even no observation at all) may be observed at the end of the experiment.  
To overcome this 
problem, Cho, Sun and Lee (\textit{Statistical Methodology}, 23, 18-34, 2015) recently introduced generalized progressive 
censoring which ensures to get a pre specified number of failures.  In this paper we analyze generalized progressive censored 
data in presence of competing risks. 
For brevity we have considered only two competing causes of failures, and it is assumed that the lifetime of the competing causes
follow one parameter exponential distributions with different scale parameters.  We obtain the maximum likelihood estimators 
of the unknown parameters and also provide their exact distributions.  Based on the exact distributions of the maximum likelihood 
estimators exact confidence intervals can be obtained.  Asymptotic and bootstrap confidence intervals are also provided 
for comparison purposes.   We further consider the Bayesian analysis of the unknown parameters under a very flexible Beta-Gamma 
prior.   We provide the Bayes estimates and the associated credible intervals of the unknown parameters based on the above 
priors.   We present extensive simulation results to see the effectiveness of the proposed method and finally one real data set 
is analyzed for illustrative purpose.      
\end{abstract}
\noindent {\sc Key Words and Phrases:} Competing risk; generalized progressive hybrid censoring; beta-gamma distribution; 
maximum likelihood estimator; bootstrap confidence interval; Bayes credible interval.

\noindent {\sc AMS Subject Classifications:} 62F10, 62F03, 62H12.


\section{\sc Introduction}

In many life testing studies, an item or individual may fail due to different  causes.  In these studies, one observes 
the time of failure along with it the corresponding cause of failure also.  These causes as if compete with each other for the 
failure of an experimental  unit.  Hence, in the statistical literature, this is known as the competing risk problem and it has been 
studied quite extensively by several researchers see for example Kalbfleish and Prentice \cite{kalbfleish}, Lawless \cite{lawless}
and the references cited therein.  In a 
competing risk problem one is naturally interested to find the lifetime distribution of the individual cause, in presence of the 
other causes.  Analysis of the competing risk data is  mainly based on one of the two assumptions namely (i) the latent failure time model 
assumptions of Cox \cite{cox} or (ii) the cause specific hazard function assumptions by Prentice et al. \cite{prentice}.  Interested 
readers may refer to the monograph of Crowder \cite{crowder} for different interesting competing risks problems and for  details
on different competing risk models.

In life testing experiment most of times the data are censored.  The most common censoring schemes used in practice are Type-I and 
Type-II censoring schemes.  Mixture of Type-I and Type-II censoring schemes led to Type-I hybrid censoring scheme which was 
introduced by Epstein \cite{epstein} and the Type-II hybrid censoring scheme introduced by Childs et al. \cite{childs}.  Hybrid 
censoring schemes have received considerable attention in the last few years, see for example the recent review article by 
Balakrishnan and Kundu \cite{balakrishnan}  regarding the recent development on this topic.  

Note that the censoring schemes outlined above do not allow for removal of units other
than the terminal points of the experiments.   Cohen \cite{cohen} was the first one to study more general censoring scheme known as
Progressive (Type-II) censoring scheme.  The Progressive censoring scheme can be described as follows.  For a fixed sample 
size $n$ and for fixed effective sample size $m$, choose non-negative integers $R_1, \ldots, R_m$ such that $\displaystyle R_1+\ldots+R_m = n-m$.  
The experiment starts with $n$ number of items.  At the time of the first failure, $R_1$ items are chosen at random from the remaining 
$n-1$ items and they are removed from the experiment.  Similarly, at the time of the second failure, $R_2$ items are chosen at 
random from the remaining
$n-R_1-2$ items and they are removed from the experiment, and the process continues.  Finally, at the time of the $m$-th failure all the 
remaining items are removed from the system and the experiment stops.  Extensive work has been done on different aspects of progressive
censoring, see for example the recent book by Balakrishnan and Cramer \cite{BC:2014} in this respect.

One major drawback of the progressive Type-II censoring scheme is that the experiment takes a longer time to continue if 
the units are highly reliable.  To overcome that problem, Kundu and Joarder \cite{joarder} introduced the progressive Type-II 
hybrid censoring scheme, where a prefixed time point $T$ is introduced along with $n, m$ and $R_1, R_2, \ldots, R_m$.  Here the experiment 
set up is same as that of the progressive type-II censoring scheme except at the termination time point.  In this case, the 
experiment stops at the time point $\min\{T, X_{m:m:n}\}$, where $X_{m:m:n}$ denotes the failure time of the $m$-th unit.  Thus in 
this set up the experiment never goes beyond the time point $T$.  
Kundu and Joarder \cite{joarder} assumed the underlying lifetime distribution to be exponential and developed the inference 
procedures of the unknown parameters.  It was further studied  
for other lifetime distributions by several authors, see for example Chan \cite{chan}, Hemmati \cite{hemmati}, Cramer 
\cite{cramer} and the referenced 
therein.  However in this censoring scheme, the experimenter may observe only few failures or in worst case no failure at all. 
This will lead to facing the problem of estimating the parameters of the lifetime distributions of the experimental units 
efficiently.

Recently Cho et al. \cite{cho} introduced a new censoring scheme called generalized progressive hybrid censoring.  It can be 
described as follows.  The experiment starts with $n$ number of items with $k$ and $m$ prefixed integers such that 
$1 \leq k < m \leq n$. The experimenter chooses predefined time point $T$ and non negative integers $R_1, R_2, \ldots, R_m$ such 
that $m+R_1+R_2+ \ldots +R_m=n$.  If $Z_{k:m:n}$ denotes the failure of the $k$-th experimental unit then the 
experiment stops at the time point $T^*=max \{Z_{k:m:n}, min \{T, Z_{m:m:n}\}\}$.  
In this paper, we develop the statistical inference of the unknown parameters based on the data coming from a 
generalized hybrid censoring scheme in presence of competing risk.  It is assumed that the lifetime distributions of the competing
causes of failures satisfy the latent failure time model assumptions of Cox \cite{cox}.  It is further assumed that there are only
two competing causes, and the lifetime distributions of the competing causes follow one parameter exponential distribution 
with mean $\theta_1$ and $\theta_2$, respectively, and they are independently distributed.   We obtain the 
maximum likelihood estimators (MLEs) of the unknown parameters and the exact distributions of the MLEs.  Based 
on the exact distribution, we compute the exact 100(1-$\alpha$)\% confidence intervals of $\theta_1$ and $\theta_2$.  For comparison
purposes we compute the bootstrap confidence intervals and the confidence intervals based on the asymptotic distributions of the 
MLEs.  We further consider the Bayesian inference of the unknown parameters based on a very general Beta-Gamma prior.  The Bayes
estimates and the corresponding credible intervals are also constructed.  Extensive simulations are performed to compare the 
different methods, and one real data set has been analyzed for illustrative purposes.

Rest of the paper is organized as follows. In Section 2, we provide the notations, model description and the MLEs.  The exact 
conditional distributions of the MLEs are provided in Section 3.  In Section 4, we provide different confidence intervals 
of the unknown parameters.  
The Bayesian inferences are provided in Section 5.  In Section 6 we provide simulations results and the analysis of a real data
set.  Finally we conclude the paper in Section 7.  

\section{\sc Notations, Model descriptions and MLEs}
\subsection{\sc Notations}
\begin{eqnarray*}
\hbox{PDF}: & & \hbox{probability density function.}   \\
X_j: & & \hbox{random variable associated with the $j$-th cause,} \ j=1, 2.\\
Z : & & \min\{X_1, X_2\}.\\
Z_{i:m:n} : & & \hbox{$i$-th observed failure of the system,}~ 1\leq i \leq m.\\
T: & & \hbox{prefixed time point.}\\
T^*: & & \max\{Z_{k:m:n}, \min\{Z_{m:m:n}, T\}\}.\\
D_j: & & \hbox{number of failures due to cause $j$ for $j = 1, 2$}.\\
J: & & \hbox{number of failures till the time point $T^*$}.\\
R_{i}= & & \hbox{number of units removed at the time of $i$-th failure}, 1\leq i \leq m.\\
R^*_k= & & n-k-\sum_{i=1}^{k-1} R_i.\\
R^*_J= & & n-J-\sum_{i=1}^J R_i.\\
\frac{1}{\theta}= & & \frac{1}{\theta_1}+\frac{1}{\theta_2}.\\
\gamma_v= & & n-v+1\sum_{u=1}^{v-1}R_u.\\
\hbox{Exp}(\theta): & & \hbox{exponential distribution with PDF:}~ \frac{1}{\theta} e^{-\frac{x}{\theta}}; ~ x>0.\\
\hbox{Gamma}(a,b): & & \hbox{Gamma distribution with PDF:}~ \frac{b^a}{\Gamma (a)} e^{-bx} x^{a-1}; \ x>0.\\
\hbox{Beta}(a,b): & & \hbox{Beta distribution with PDF:}~ \frac{\Gamma(a+b)}{\Gamma (a) \Gamma (b)} x^{a-1} (1-x)^{b-1}; \ 0<x<1.\\
f_G(x;a,b,c): & & \hbox{shifted Gamma distribution with PDF:}~ \frac{c^{b}}{\Gamma(b)} (x-a)^{b-1} e^{-c(x-a)}; ~ x > a. \\  
\end{eqnarray*}

\subsection{\sc Model descriptions, MLEs}

Let $X_1$ and $X_2$ be the random variables denoting life time distributions of cause 1 and cause 2, respectively.  It is 
assumed that $X_i \sim \hbox{Exp}(\theta_i)$ for $ i=1,2$ and they are independently distributed.   At any failure time point 
one observes $Z = \min\{X_1,X_2\}$ and the associated cause of failure. We define a new indicator variable $\delta$ with $\delta_i=1$  if the $i$-th failure happens due to Cause-1 and $\delta_i=0$  if the $i$-th failure happens due to Cause-2. The probability density function of $Z$ has the following form,
\begin{equation}
f(z;\theta_1,\theta_2)=\begin{cases}
\bigg(\frac{1}{\theta_1}+\frac{1}{\theta_2}\bigg)e^{-z\big(\frac{1}{\theta_1}+\frac{1}{\theta_2}\big)}; & \text{if $z>0$},\\
0; & \text{otherwise}.
\end{cases}
\end{equation}
Under the generalized progressive hybrid censoring scheme, 
the possible values of $T^*$ are
$$
T^*=
\begin{cases}
T & \text{if \ \ \ $Z_{k:m:n}<T<Z_{m:m:n}$,}  \\
Z_{k:m:n} & \text{if \ \ \ $T<Z_{k:m:n}<Z_{m:m:n}$,}  \\
Z_{m:m:n} & \text{if \ \ \ $Z_{k:m:n}<Z_{m:m:n}<T$}.  \\
\end{cases}
$$
It is to be noted that, the likelihood contribution at the point $(z,\delta=1)$ is given by,
$$
L(\theta_1,\theta_2|(z,\delta=1))=\frac{1}{\theta_1}e^{-\frac{1}{\theta_1} z} e^{-\frac{1}{\theta_2} z}=\frac{1}{\theta_1}e^{-(\frac{1}{\theta_1}+\frac{1}{\theta_2})z}.
$$
Similarly, the likelihood contribution at the point $(z,\delta=0)$ is given by,
$$
L(\theta_1,\theta_2|(z,\delta=0))=\frac{1}{\theta_2}e^{-\frac{1}{\theta_2} z} e^{-\frac{1}{\theta_1} z}=\frac{1}{\theta_2}e^{-(\frac{1}{\theta_1}+\frac{1}{\theta_2})z}.
$$
Hence, based on the observations, the likelihood function can be written as,
\begin{equation*}
L(\theta_1, \theta_2|Data)= c \Big(\frac{1}{\theta_1}\Big)^{D_1} \Big(\frac{1}{\theta_2}\Big)^{D_2} e^{-\frac{W}{\theta}},
\end{equation*}
where,
$$
c=\begin{cases}
\prod_{v=1}^J \gamma_v, & \text{if $T^*=T$,}\\
\prod_{v=1}^k \gamma_v, & \text{if $T^*=Z_{k:m:n}$,}\\
\prod_{v=1}^m \gamma_v, & \text{if $T^*=Z_{m:m:n}$,}\\
\end{cases}
$$
$$
\gamma_v=n-v+1\sum_{u=1}^{v-1}R_u, ~~~~~ v \in \{1,2,\ldots,m\},
$$
$$
W=\begin{cases}
\sum_{i=1}^{J} z_i (1+R_i) + T R^*_J, & \text{if $T^*=T$,}\\
\sum_{i=1}^{k-1} z_i (1+R_i) + z_k (1+R^*_k), & \text{if $T^*=Z_{k:m:n}$,}\\
\sum_{i=1}^{m} z_i (1+R_i), & \text{if $T^*=Z_{m:m:n}$.}\\
\end{cases}
$$
\noindent Note that, $\gamma_v$ denotes the number of units remaining at the time of $v$-th failure.
\noindent Clearly, $J=k,k+1,\ldots,m-1$ when $T^*=T$. Thus the log likelihood function (ignoring the constant $c$) is given by,
\begin{equation}\label{loglikelihood}
l(\theta_1, \theta_2)= -D_1 ~\log ~ \theta_1 - D_2~ \log ~ \theta_2 - \frac{W}{\theta}.
\end{equation}
Note that the MLE of $\theta_1$ ($\theta_2$) exists if $D_1 >0$ ($D_2>0$).  If $D_1=0$ ($D_2$ = 0) then the likelihood function 
is unbounded above as a function of $\theta_1$ ($\theta_2$) and thus the MLE of $\theta_1$ ($\theta_2$) does not exist.  
Taking partial derivatives of equation (\ref{loglikelihood}) with respect to $\theta_1$ and $\theta_2$ and equating them to $0$ 
we get MLEs of $\theta_1$ (given, $D_1>0$) and $\theta_2$ (given, $D_2>0$) as
\begin{equation}
\widehat{\theta}_1=\frac{W}{D_1} \ \ \hbox{and} \ \ \widehat{\theta}_2=\frac{W}{D_2}.
\end{equation}
\section{\sc Main Result: Exact Conditional distributions of the MLEs}
In this section we provide the exact marginal distributions of the MLEs of the parameters $\theta_1$ and $\theta_2$ 
conditioning on $D_1>0$ and $D_2>0$, respectively.  
\enlargethispage{1 cm}
\begin{theorem} \label{exact_theta_1} The conditional PDF of $\widehat{\theta}_1$ conditioning on $D_1>0$, is given by
\begin{align*}
f_{\widehat{\theta}_1|D_1>0}(x)&=\frac{1}{P(D_1>0)}\Bigg[\sum_{j=k}^{m-1} \sum_{i=1}^{j}\Bigg\{ \prod_{v=1}^j \gamma_v \bigg(\frac{1}{\theta_1}\bigg)^i\bigg(\frac{1}{\theta_2}\bigg)^{j-i} \Big(\frac{1}{\theta}\Big)^{-j}\\
& \quad \sum_{v=0}^j \frac{(-1)^v e^{-\frac{T}{\theta}\gamma_{j-v+1}}}{\{\prod_{h=1}^v (\gamma_{j+1-v}-\gamma_{j+1-v+h})\}\{\prod_{h=1}^{j-v}  (\gamma_h-\gamma_{j-v+1})\}}\\
& \quad f_G\Big(x;\frac{T}{i}\gamma_{j-v+1},j,\frac{i}{\theta}\Big)\Bigg\}\\
& \quad +\sum\nolimits_{i=1}^k \Bigg\{\prod_{v=1}^k \gamma_v \bigg(\frac{1}{\theta_1}\bigg)^i\bigg(\frac{1}{\theta_2}\bigg)^{k-i} \Big(\frac{1}{\theta}\Big)^{-k}\\
& \quad \sum_{v=0}^{k-1} \frac{(-1)^ve^{-\frac{T}{\theta}\gamma_{k-v}}}{\{\prod_{j=1}^v (\gamma_{k-v}-\gamma_{k-v+j})\}\{\prod_{j=1}^{k-1-v} (\gamma_j-\gamma_{k-v})\}\gamma_{k-v}}\\
& \quad  f_G\Big(x;\frac{T}{i}\gamma_{k-v},k,\frac{i}{\theta}\Big)\Bigg\}\\
& \quad+\sum\nolimits_{i=1}^m \Bigg\{\prod_{v=1}^m \gamma_{v} \bigg(\frac{1}{\theta_1}\bigg)^i\bigg(\frac{1}{\theta_2}\bigg)^{m-i} \Big(\frac{1}{\theta}\Big)^{-m}\\
& \quad \sum_{v=0}^{m} \frac{(-1)^v e^{-\frac{T}{\theta}(\gamma_{m-v+1}-\gamma_{m+1})}}{\{\prod_{j=1}^v  (\gamma_{m-v+1}-\gamma_{m-v+j+1})\}\{\prod_{j=1}^{m-v} (\gamma_j-\gamma_{m-v+1})\}}\\
& \quad f_G\Big(x;\frac{T}{i}(\gamma_{m-v+1}-\gamma_{m+1}),m,\frac{i}{\theta}\Big)\Bigg\}\Bigg],
\end{align*}
\end{theorem}
\noindent where,
\begin{align*}
P(D_1>0)&=1-P(D_1=0)
\end{align*}
and
\begin{align*}
P(D_1=0)=&\sum_{j=k}^m \prod_{v=1}^{j+1} \sum_{u=0}^{j}\frac{\gamma_v(-1)^u e^{-\frac{T}{\theta}\gamma_{j-u+1}}}{\{\prod_{v=1}^u  (\gamma_{j+1-u}-\gamma_{j+1-u+v})\}\{\prod_{v=1}^{j-u} (\gamma_v-\gamma_{j-u+1})\}\{\gamma_{j-u+1}-u\}}\Big(\frac{\theta_1}{\theta_1+\theta_2}\Big)^j\\
&+\prod_{v=1}^k \sum_{v=0}^{k-1} \frac{\gamma_v(-1)^v e^{-\frac{T}{\theta}\gamma_{k-v}}}{\{\prod_{j=1}^v  (\gamma_{k-v}-\gamma_{k-v+j})\}\{\prod_{j=1}^{k-1-v}  (\gamma_j-\gamma_{k-v})\}\{\gamma_{k-v}\}}\Big(\frac{\theta_1}{\theta_1+\theta_2}\Big)^k \\
&+\prod_{v=1}^m \sum_{v=0}^{m} \frac{\gamma_v(-1)^v e^{-\frac{T}{\theta}(\gamma_{m-v+1}-\gamma_{m+1})}}{\{\prod_{j=1}^v  (\gamma_{m-v+1}-\gamma_{m-v+j+1})\}\{\prod_{j=1}^{m-v} (\gamma_j-\gamma_{m-v+1})\}}\Big(\frac{\theta_1}{\theta_1+\theta_2}\Big)^m.\\
\end{align*}
\begin{proof}
See in the Appendix.
\end{proof}

\noindent Similarly, the conditional distribution of $\widehat{\theta}_2$ given $D_2>0$ is obtained by replacing $D_1$ and  $\theta_1$ 
by $D_2$ and $\theta_2$ respectively.

\noindent {\bf Comment:} The conditional PDFs of $\widehat{\theta}_1$ and $\widehat{\theta}_2$ are quite complicated.  To get
some idea about the shape of the PDFs, we plot the conditional PDF of $\widehat{\theta}_1$ and $\widehat{\theta}_2$ in Figure
\ref{pdf-theta-1} and Figure \ref{pdf-theta-2}, respectively for $n$ = 20, $k$ = 5, $m$ = 18, $T$ = 1.2, $\theta_1$ = 1.0, 
$\theta_2$ = 1.3, and for the censoring Scheme-III (see Section 6 for details).  We have also plotted the histograms of 
$\widehat{\theta}_1$ and $\widehat{\theta}_2$ based on 5000 replications on the same graph, and they match very well.  It verifies
the correctness of the expressions of the conditional distributions of $\widehat{\theta}_1$ and $\widehat{\theta}_2$.

\begin{figure}
\includegraphics[scale=0.5]{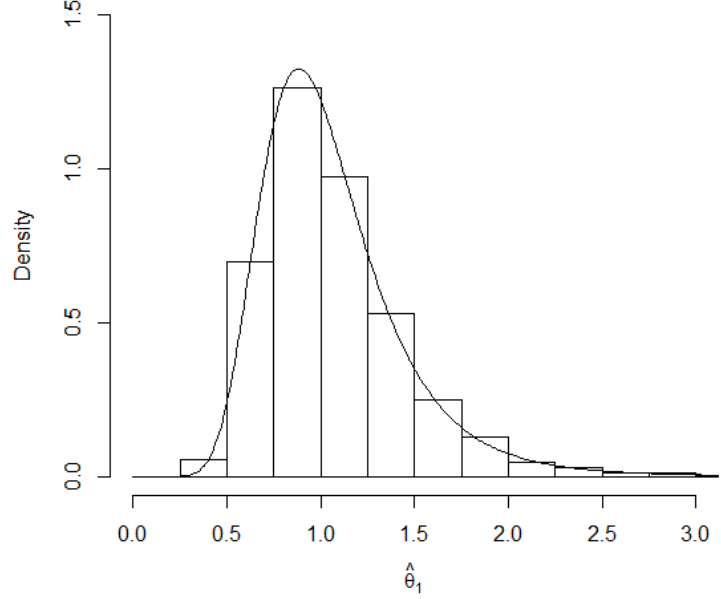}
\centering
\captionsetup{font=scriptsize}
\caption{Histogram of $\widehat{\theta}_1$ along its PDF    \label{pdf-theta-1}}
\end{figure}

\begin{figure}
\includegraphics[scale=0.5]{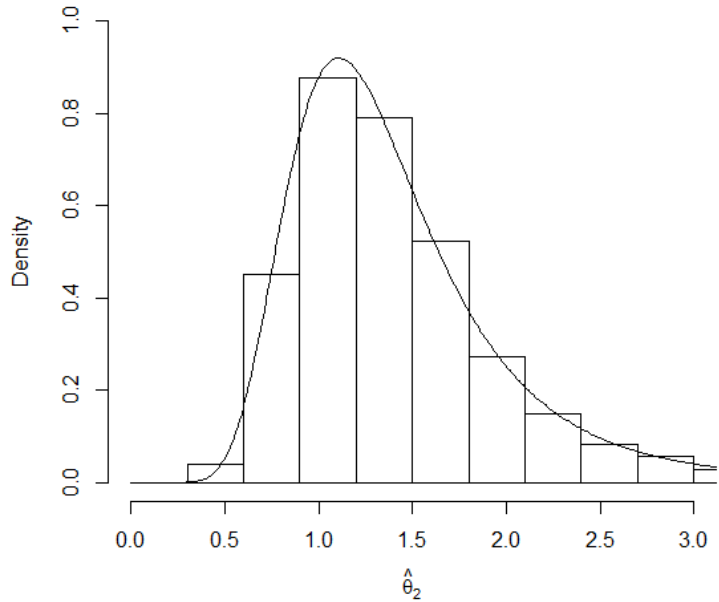}
\centering
\captionsetup{font=scriptsize}
\caption{Histogram of $\widehat{\theta}_2$ along its PDF   \label{pdf-theta-2}}
\end{figure}

\section{\sc Confidence interval}

From the exact conditional distribution of $\widehat{\theta}_1$, the exact CI can be constructed under the assumption 
that $\displaystyle P_{\theta_1}(\widehat{\theta}_1 \leq x)$ is a monotonic decreasing function of $\theta_1$ for fixed $x$.  Several 
authors including Balakrishnan et al. \cite{xie}, Chen and Bhattayacharya \cite{chen}, Childs et al. \cite{childs}, 
Kundu and Basu \cite{basu} used this technique to construct CI for the parameters.  Although we cannot prove the 
monotonicity property of $\displaystyle P_{\theta_1}(\widehat{\theta}_1 \leq x)$ analytically, a graphical plot supports this 
property, see Figure \ref{fig-1} and Figure \ref{fig-2}.
\begin{figure}
\includegraphics[scale=0.5]{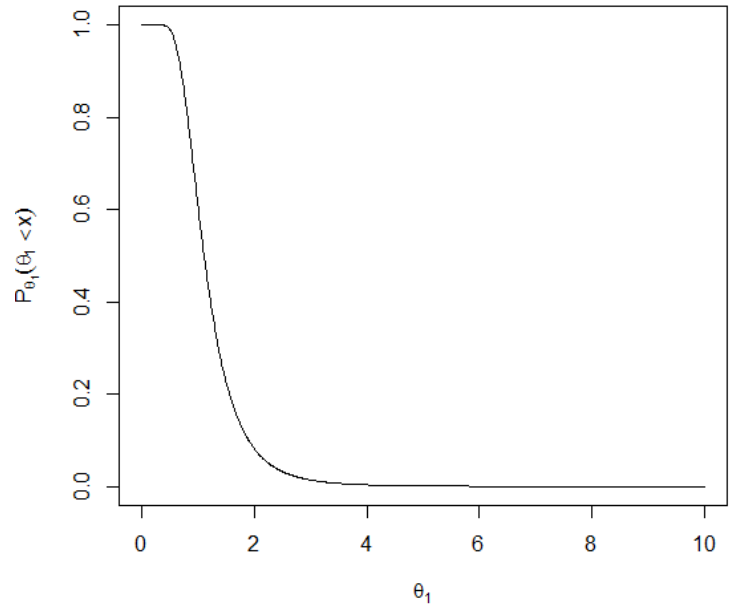}
\centering
\captionsetup{font=scriptsize}
\caption{The plot of $P_{\theta_1}(\widehat{\theta}_1 \le x)$ for $ n=20, k=3, m=14, T=1.2, x=1.109, \theta_2=1.472$ under Scheme-I   \label{fig-1}}
\end{figure}
\begin{figure}
\includegraphics[scale=0.5]{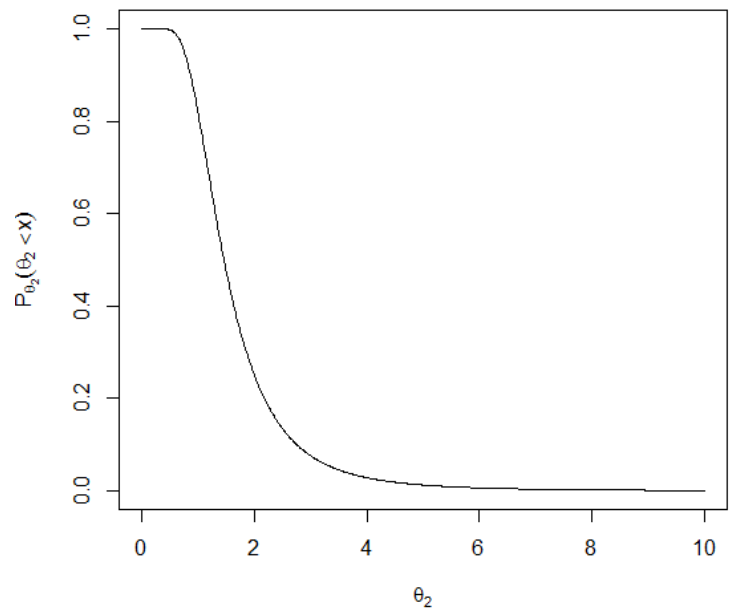}
\centering
\captionsetup{font=scriptsize}
\caption{The plot of $P_{\theta_2}(\widehat{\theta}_2 \le x)$ for $ n=20, k=3, m=14, T=1.2, x=1.472, \theta_1=1.109$ under Scheme-I  \label{fig-2}}
\end{figure}

The confidence interval of $\theta_1$ can be obtained as follows, see Kundu and Basu \cite{basu} for details.
Suppose, $\widehat{\theta}_{1,obs}$ is an estimated value of $\theta_1$. Then for $\alpha \in (0,1)$, a $100(1-\alpha)\%$ exact CI is given by $(\theta_{1l}, \theta_{1u})$ where, $\theta_{1l}$ and $\theta_{1u}$ are obtained from the following two equations,
\begin{equation}
P_{\theta_{1l}}(\widehat{\theta}_1 \leq \widehat{\theta}_{1,obs})=1-\frac{\alpha}{2} \ \ \ \ \hbox{and} \ \ \ \
P_{\theta_{1u}}(\widehat{\theta}_1 \leq \widehat{\theta}_{1,obs})=\frac{\alpha}{2}.
\end{equation}
Clearly due to complicated nature of PDF of $\widehat{\theta}_1$, the above two equations turn out to be non linear 
equations.  One needs to solve them numerically, say for example by Newton Rapshon method or Bisection method.   
Similarly we can obtain the confidence interval for $\theta_2$ also.  

It should be mentioned that the exact confidence intervals 
of $\theta_1$ and $\theta_2$ may not always exist, see for example Balakrishnan et al. \cite{BCI:2014} for all $0 < \alpha < 1$.  In fact it is clear that if 
$\displaystyle P_{\theta_1}(\widehat{\theta}_1 \le \widehat{\theta}_{1,obs})$ varies from 1 to 0, as $\theta_1$ varies from 0 to infinity, for
all values of $\widehat{\theta}_{1,obs}$, then the exact confidence interval of $\theta_1$ exists for all values of $0 < \alpha < 1$.
Similarly, for $\theta_2$ also.  It is not difficult to show that $\displaystyle
\lim_{\theta_1 \rightarrow 0} P_{\theta_1}(\widehat{\theta}_1 \le \widehat{\theta}_{1,obs}) = 1$ and $\displaystyle
\lim_{\theta_2 \rightarrow 0} P_{\theta_2}(\widehat{\theta}_2 \le \widehat{\theta}_{2,obs}) = 1.$
But we could not establish the following limits $\displaystyle 
\lim_{\theta_1 \rightarrow \infty} P_{\theta_1}(\widehat{\theta}_1 \le \widehat{\theta}_{1,obs})$ and 
$\displaystyle \lim_{\theta_2 \rightarrow \infty} P_{\theta_2}(\widehat{\theta}_2 \le \widehat{\theta}_{2,obs})$.
In fact it has been observed in our simulation experiments that these limits may not be 0 always, and in those cases the exact finite length
confidence intervals do not exist for all $0 < \alpha < 1$.  We have demonstrated that 
in Section 6.

Due to complicated functional form of the exact conditional distribution function, it is quite difficult to obtain exact confidence intervals of the parameters in practice.  Hence, we propose to use bootstrap method for construction of CIs of the unknown parameters. The steps are similar as given in Kundu and Gupta \cite{gupta}, hence they are omitted.

\section{\sc Bayesian Analysis}

In this section we consider the Bayesian inference of the unknown parameters.  We obtain the Bayes estimates and the highest 
posterior density (HPD) credible intervals of the unknown parameters.  Suppose,  
$\displaystyle \lambda_1=\frac{1}{\theta_1}$ and $\displaystyle \lambda_2=\frac{1}{\theta_2}$, 
we take a joint prior distribution on $(\lambda_1, \lambda_2)$ as the Beta-gamma distribution, 
see for example Pena and Gupta \cite{pena}, with parameters $b_0>0, a_0>0, a_1>0, a_2>0$.  From now on it will be denoted by 
BG$(b_0,a_0,a_1,a_2)$ and it has the joint PDF for $\lambda_1, \lambda_2 > 0$ as follows;
\begin{align}
\pi(\lambda_1, \lambda_2|b_0, a_0, a_1, a_2)=\frac{\Gamma(a_1+a_2)}{\Gamma(a_0)}(b_0(\lambda_1+\lambda_2))^{a_0-a_1-a_2} \frac{b_0^{a_1}}{\Gamma(a_1)} \lambda_1^{a_1-1} e^{-b_0\lambda_1} \frac{b_0^{a_2}}{\Gamma(a_2)} \lambda_2^{a_2-1} e^{-b_0\lambda_2}.   \label{prior}
\end{align}
The joint posterior distribution of $(\lambda_1, \lambda_2)$ is given by,
\begin{align*}
\pi(\lambda_1, \lambda_2|data) \propto e^{-(W+b_0)(\lambda_1+\lambda_2)} \lambda_1^{a_1+D_1-1} \lambda_2^{a_2+D_2-1} (\lambda_1+\lambda_2)^{a_0-a_1-a_2};
\ \ \ \ \lambda_1 > 0, \lambda_2 > 0,
\end{align*}
and it is BG$(W+b_0,a_0+D_1+D_2,a_1+D_1,a_2+D_2)$.
The Bayes estimates of $\theta_1$ and $\theta_2$ under square error loss function are, respectively,
\begin{align*}
\widehat{\theta}_{1B} & = E\Big(\frac{1}{\lambda_1}|data\Big)=\frac{(a_1+a_2+D_1+D_2-1)(w+b_0)}{(a_0+D_1+D_2-1)(a_1+D_1-1)}
\end{align*}
and
\begin{align*}
\widehat{\theta}_{2B}=E\Big(\frac{1}{\lambda_2}|data\Big)=\frac{(a_1+a_2+D_1+D_2-1)(w+b_0)}{(a_0+D_1+D_2-1)(a_2+D_2-1)}.
\end{align*}
The corresponding posterior variances are, respectively, 
\begin{align*}
V(\theta_1|Data)=A_1B_1 \ \ \text{and} \ \ \ V(\theta_2|Data)=A_2B_2,
\end{align*}
where for $k=1,2$,
$$
A_k=\frac{(w+b_0)^2(a_1+a_2+D_1+D_2-1)}{(a_0+D_1+D_2-1)(a_k+D_k-1)}
$$
and
$$
B_k=\frac{(a_1+a_2+D_1+D_2-2)}{(a_0+D_1+D_2-2)(a_k+D_k-2)}-\frac{(a_1+a_2+D_1+D_2-1)}{(a_0+D_1+D_2-1)(a+k+D_k-1)}.
$$
\subsection{\sc credible set}
Let us recall that for $\alpha \in (0,1)$, a set $C_\alpha$ is said to be a $100(1-\alpha)\%$ credible set for ($\lambda_1,\lambda_2$) 
if $P[(\lambda_1,\lambda_2)\in C_\alpha]=1-\alpha$, where ($\lambda_1,\lambda_2$) $\sim \pi(\lambda_1,\lambda_2|data)$.
To construct a  $100(1-\alpha)\%$ credible set we use the following theorem.  The proof is quite straightforward, hence it is not 
provided here.
\begin{theorem}
If $(\lambda_1,\lambda_2)\sim BG(b_0+W, a_0+J, a_1+D_1, a_2+D_2)$ then,
$$
U=\lambda_1+\lambda_2 \sim Gamma(a_0+D_1+D_2, b_0+W), \ \ \ \ V=\frac{\lambda_1}{\lambda_1+\lambda_2} \sim Beta (a_1+D_1, a_2+D_2)
$$ 
and they are independently distributed.
\end{theorem} 
Using the above theorem we construct a  $100(1-\alpha)\%$ credible set of ($\lambda_1,\lambda_2$) as follows.
Let $0<\alpha_1, \alpha_2<1$ such that $(1-\alpha_1)(1-\alpha_2)=1-\alpha$. Then $C_\alpha$ can be constructed as,
$$
C_{\lambda_1,\lambda_2,\alpha}=\{(\lambda_1, \lambda_2): A<\lambda_1+\lambda_2<B, C<\frac{\lambda_1}{\lambda_1+\lambda_2}<D\}.
$$ 
with 
$\displaystyle P(A<\lambda_1+\lambda_2<B)=1-\alpha_1$ and $\displaystyle P(C<\frac{\lambda_1}{\lambda_1+\lambda_2}<D)=1-\alpha_2$.

\noindent Note that, $C_{\lambda_1,\lambda_2,\alpha}$ is a trapezoid enclosed by the following four straight lines:
\begin{equation}
(i)\ \lambda_1+\lambda_2 = A, \ \ (ii) \ \lambda_1+\lambda_2 = B, \ \ (iii) \ \lambda_1(1-D) = \lambda_2 D, \ \ (iv) \ \lambda_1(1-C) 
= \lambda_2 C.  
\end{equation}
Equivalently, we can get a 100$(1-\alpha)\%$ credible set of $(\theta_1, \theta_2)$ as follows
$$
C_{\theta_1,\theta_2,\alpha}=\{(\theta_1, \theta_2): A< \frac{\theta_1+\theta_2}{\theta_1 \theta_2} <B, C<\frac{\theta_2}{\theta_1+\theta_2}<D\}.
$$ 
Hence, $\displaystyle C_{\theta_1,\theta_2,\alpha}$ is an area enclosed by the following four lines (two of them are straight lines and two of them are
curved lines).
$$
(i)\ \theta_1+\theta_2 = A \theta_1 \theta_2, \ \ (ii) \ \theta_1+\theta_2 = B \theta_1 \theta_2, \ \ 
(iii) \ \theta_1 D = \theta_2 (1- D), \ \ (iv) \ \theta_1 C = \theta_2 (1- C).  
$$
Although, we can obtain a $100(1-\alpha)\%$ credible set of $(\lambda_1, \lambda_2)$ or of $(\theta_1, \theta_2)$ by the above method, 
it is not possible to obtain the HPD credible intervals for any function of $\lambda_1$, $\lambda_2$ or of $\theta_1, \theta_2$. 
  We propose to use the Gibbs sampling technique to construct the $100(1-\alpha)\%$ HPD credible interval 
of any function of $\lambda_1$ and $\lambda_2$, say, $g(\lambda_1, \lambda_2)$.  Shrijita \textit{et al.} \cite{shrijita} gave similar algorithm to find HPD credible interval of $g(\lambda_1,\lambda_2)$ and hence are omitted here.

\section{\sc Simulation and Data Analysis}
\subsection{\sc Simulation}

In this section we present some simulation results mainly to see how the different methods proposed in this paper behave in 
practice.  We have kept $\theta_1 = 1$, $\theta_2 = 1.3$, $T$ = 1.2 fixed.  We have considered different sample sizes ($n$), 
different effective sample sizes ($m$), different $k$ values and three different progressive censoring schemes namely; 
Scheme-I ($R_1 = n-m$, $R_2 = \ldots = R_m$ = 0),  Scheme-II ($R_{m/2}$ = 0, $R_1 = \ldots = R_{m/2-1} = R_{m/2+1} = \ldots = R_m$ = 0)
and Scheme-III ($R_1 = \ldots = R_{m-1} = 0, R_m = n-m$).  We have considered the following different cases, see Table \ref{cases}.
\begin{table}[h!]
\centering
\caption{Different cases considered for simulation experiments.   \label{cases}}   

\vspace{0.25 in}
\scalebox{0.8}
{
\begin{tabular}{|c| c| c|}
\hline
Scheme-I	&	Scheme-II	&	Scheme-III\\
\hline
n=20, k=3, m=14	&	n=20, k=3, m=14	&	n=20, k=3, m=14\\
n=20, k=5, m=14	&	n=20, k=5, m=14	&	n=20, k=5, m=14\\
n=20, k=3, m=18	&	n=20, k=3, m=18	&	n=20, k=3, m=18\\
n=20, k=5, m=18	&	n=20, k=5, m=18	&	n=20, k=5, m=18\\
n=30, k=6, m=22	&	n=30, k=6, m=22	&	n=30, k=6, m=22\\
n=30, k=8, m=22	&	n=30, k=8, m=22	&	n=30, k=8, m=22\\
n=30, k=6, m=26	&	n=30, k=6, m=26	&	n=30, k=6, m=26\\
n=30, k=8, m=26	&	n=30, k=8, m=26	&	n=30, k=8, m=26\\
\hline
\end{tabular}
}
\end{table}       
In each case we compute the MLEs of the unknown parameters, and construct 
the 95 \% exact confidence intervals and the bootstrap confidence intervals  
for both the parameters.  We replicate the process 5000 times, and report the average biases, mean squared errors (MSE), the 
average length of the confidence intervals and the associated coverage percentages within brackets. 

One important point we would like to mention that a finite length of 95 \% exact confidence interval of each 
of the parameters does not exist for each replication. This happens when the number of failures due to a particular cause is 
very low (1 or 2 mostly).   We have reported the number of times the exact confidence intervals do not exist (NECI) also.  The results are 
presented in 
Table \ref{freq-theta-1} and Table \ref{freq-theta-2}.
We give few graphical evidences of such non existence of solutions [see Figure \ref{fig-7} and Figure \ref{fig-8}] .  

To perform the Bayesian analysis we consider two different priors; Prior-I (informative) and Prior-II (matching prior).  For 
Prior-I, the hyper parameters are $\displaystyle a_0 =46/13, b_0 = 1, a_1 = 2.3, a_2 = 2$, and they have been chosen in such a
manner that $E(\lambda_1) = 1$ and $E(\lambda_2) = 1/1.3$, the true values of the parameters considered for simulation experiments. 
For Prior-II, the hyper parameters are $a_0=2, b_0=0, a_1=1, a_2=1$, and they have been chosen in such a manner that the Bayes 
estimates match with the corresponding MLEs.  In this case also similarly as before, in each case we compute the Bayes estimates 
and the symmetric and HPD credible intervals.  We replicate the process 5000 times, and report the average biases, MSEs, the average
lengths of the credible intervals and the associated coverage percentages within brackets.  The results are presented in 
Table  \ref{informative-theta-1}, Table  \ref{non-informative-theta-1}, Table  \ref{informative-theta-2} and 
Table  \ref{non-informative-theta-2}.   

Some of the points are quite clear from these simulation results.  It is observed that the performance of the MLEs are quite 
satisfactory.  As the sample size increase the biases and the MSEs decrease as expected.  Comparing the performances of the 
different confidence intervals it is observed that the performances of the exact confidence intervals and the bootstrap confidence
intervals are quite good.  In both these cases coverage percentage are maintained.  Since the average lengths of the bootstrap 
confidence intervals are slightly smaller than the corresponding average lengths of the exact confidence intervals we propose to use
the bootstrap confidence intervals for all practical purposes, and they are very easy to implement also.

Clearly, from Figure \ref{fig-7} and Figure \ref{fig-8} we see that upper confidence limits of the parameters $\theta_1$ and $\theta_2$ do not exist for large estimates of the corresponding parameters. In other words, it can be said that finite length of confidence interval does not exist for $\theta_1$ when $D_1$ is very small. Similar case holds for $\theta_2$ also.

The performance of the Bayes estimates also are quite satisfactory.  The average biases and MSEs of the Bayes estimates under 
Prior-II match with the corresponding values associated with the MLEs, as it should be.  In case of Prior-I the average biases
and the associated MSEs are slightly smaller than Prior-II.  Regarding the credible intervals it is clear that both the symmetric
and HPD credible intervals are performing quite well.  Since the coverage percentages for the symmetric credible intervals are 
closer to the nominal value, it might be preferable compared to the HPD credible intervals.  Hence, 
if we do not have any prior information, then use the Bayes estimates with the matching prior, and 
if we have any prior information, then use the informative prior.

\begin{figure}
\includegraphics[scale=0.5]{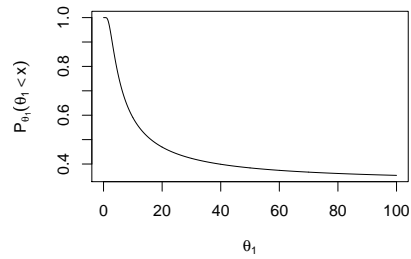}
\centering
\captionsetup{font=scriptsize}
\caption{The plot of $P_{\theta_1}(\widehat{\theta}_1 \le x)$ for $ n=20, k=3, m=14, T=1.2, x=7.549, \theta_2=0.755$ under Scheme-I   \label{fig-7}}
\end{figure}

\begin{figure}
\includegraphics[scale=0.5]{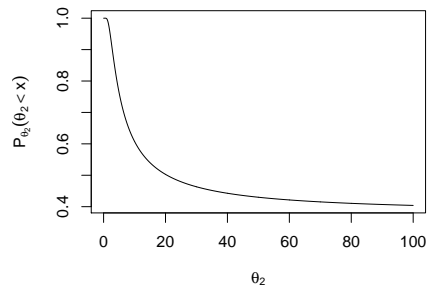}
\centering
\captionsetup{font=scriptsize}
\caption{The plot of $P_{\theta_2}(\widehat{\theta}_2 \le x)$ for $ n=20, k=3, m=14, T=1.2, x=6.864, \theta_1=0.624$ under Scheme-I   \label{fig-8}}
\end{figure}

\begin{table}[h!]
\captionsetup{font=scriptsize}
\caption {Results of frequentist analysis for $\theta_1$}
\label{freq-theta-1} \centering
\scalebox{0.65}
{
 \begin{tabular}{|c c c c c c c c c|}
\hline
n & m & k & \thead{Scheme} & \thead{Bias} & \thead{MSE} & \thead{Exact \\ CI} &	NECI	& \thead{Bootstrap \\ CI}   \\ [0.5ex]
 \hline\hline
 20 &   14  &   3   &   I   &   0.109   &   0.269   &   2.698 (96)  &	3	&   2.601 (95)  \\
    &       &       &   II  &   0.098   &   0.253   &   2.543 (95)  &	5	&   2.452 (94)  \\
    &       &       &   III &   0.068   &   0.200   &   2.173 (96)  &	3	&   2.082 (94)  \\
    &       &   5   &   I   &   0.108   &   0.255   &   2.678 (96)  &	7	&   2.604 (95)  \\
    &       &       &   II  &   0.092   &   0.248   &   2.602 (95)  &	3	&   2.425 (94)  \\
    &       &       &   III &   0.065   &   0.192   &   2.148 (95)  &	1	&   2.068 (94)  \\
    &	18	&	3	&	I	&	0.079	&	0.175	&	1.879 (96)	&	1	&	2.027 (94)	\\
    &		&		&	II	&	0.085	&	0.190	&	1.937 (96)	&	0	&	2.019 (94)	\\
    &		&		&	III	&	0.080	&	0.166	&	1.761 (95)	&	2	&	1.878 (94)	\\
    & 	    &   5   &   I   &   0.083   &   0.189   &   1.925 (95)  &	1	&   2.051 (94)  \\
    &       &       &   II  &   0.072   &   0.176   &   1.862 (95)  &	0	&   1.968 (94)  \\
    &       &       &   III &   0.066   &   0.155   &   1.726 (96)  &	0	&   1.847 (94)  \\
 30 &   22  &   6   &   I   &   0.067   &   0.145   &   1.584 (95)  &	0	&   1.681 (94)  \\
    &       &       &   II  &   0.063   &   0.138   &   1.543 (94)  & 	0	& 	1.614 (94)   \\
    &       &       &   III &   0.045   &   0.112   &   1.410 (94)  &   0	&	1.408 (93)  \\
    &       &   8   &   I   &   0.060   &   0.139   &   1.559 (95)  &   0	&	1.666 (94)   \\
    &       &       &   II  &   0.063   &   0.138   &   1.543 (94)  &   0	&	1.614 (94)    \\
    &       &       &   III &   0.045   &   0.112   &   1.410 (95)  &   0	&	1.408 (93)    \\
    &   26  &   6   &   I   &   0.053   &   0.110   &   1.371 (95)  &   0	&	1.437 (94)  \\
    &       &       &   II  &   0.051   &   0.104   &   1.349 (95)  &   0	&	1.400 (94)  \\
    &       &       &   III &   0.042   &   0.087   &   1.259 (96)  &   0	&	1.291 (95)   \\
    &       &   8   &   I   &   0.053   &   0.110   &   1.371 (95)  &   0	&	1.437 (94)  \\
    &       &       &   II  &   0.051   &   0.104   &   1.349 (95)  &   0	&	1.400 (94)  \\
    &       &       &   III &   0.042   &   0.087   &   1.259 (96)  &   0	&	1.291 (95)  \\
 \hline
 \end{tabular}
 }
\end{table}
\begin{table}[h!]
\captionsetup{font=scriptsize}
\caption {Results of Bayesian analysis under Prior-I for
$\theta_1$}
 \label{informative-theta-1} \centering
 \scalebox{0.7}
{
 \begin{tabular}{|c c c c c c c c|}
 \hline
 n & m & k & \thead{Scheme} & \thead{Bias} & \thead{MSE} & \thead{Symmetric \\ Credible Interval} & \thead{HPD \\ Credible Interval}  \\ [0.5ex]
 \hline\hline
 20 &   14  &   3   &   I   &   0.087   &   0.164   &   1.654 (96)  &   1.486 (95)\\
    &       &       &   II  &   0.080   &   0.159   &   1.616 (96)  &   1.457 (94)\\
    &       &       &   III &   0.059   &   0.135   &   1.496 (96)  &   1.360 (94)\\
    &       &   5   &   I   &   0.087   &   0.158   &   1.652 (97)  &   1.485 (95)\\
    &       &       &   II  &   0.074   &   0.153   &   1.606 (96)  &   1.447 (95) \\
    &       &       &   III &   0.057   &   0.131   &   1.490 (96)  &   1.356 (94)\\
    &	18	&	3	&	I	&	0.070	&	0.125	&	1.431 (96)	&	1.311 (95)\\
    &		&		&	II	&	0.074	&	0.131	&	1.436 (96)	&	1.316 (95)\\
    &		&		&	III	&	0.073	&	0.122	&	1.377 (96)	&	1.269 (95)\\
    &  	  	&   5   &   I   &   0.073   &   0.133   &   1.442 (96)  &   1.320 (95)\\
    &       &       &   II  &   0.064   &   0.125   &   1.414 (96)  &   1.296 (95)\\
    &       &       &   III &   0.061   &   0.116   &   1.358 (96)  &   1.252 (95)\\
 30 &   22  &   6   &   I   &   0.062   &   0.109   &   1.288 (96)  &   1.195 (94)\\
    &       &       &   II  &   0.059   &   0.106   &   1.266 (95)  &   1.176 (94)\\
    &       &       &   III &   0.045   &   0.090   &   1.178 (95)  &   1.102 (94)\\
    &       &   8   &   I   &   0.056   &   0.106   &   1.277 (95)  &   1.185 (94)  \\
    &       &       &   II  &   0.059   &   0.106   &   1.266 (95)  &   1.176 (94)  \\
    &       &       &   III &   0.045   &   0.090   &   1.178 (95)  &   1.102 (94)  \\
    &   26  &   6   &   I   &   0.052   &   0.088   &   1.168 (96)  &   1.095 (95)\\
    &       &       &   II  &   0.050   &   0.085   &   1.156 (95)  &   1.084 (95)\\
    &       &       &   III &   0.043   &   0.073   &   1.092 (96)  &   1.029 (95)\\
    &       &   8   &   I   &   0.052   &   0.088   &   1.168 (96)  &   1.095 (95)\\
    &       &       &   II  &   0.050   &   0.085   &   1.156 (95)  &   1.084 (95)\\
    &       &       &   III &   0.043   &   0.073   &   1.092 (96)  &   1.029 (95)\\
  \hline
 \end{tabular}
 }
\end{table}
\begin{table}[h!]
\captionsetup{font=scriptsize}
\caption {Results of Bayesian analysis under Prior-II for
$\theta_1$}
 \label{non-informative-theta-1} \centering
 \scalebox{0.7}
{
 \begin{tabular}{|c c c c c c c c|}
 \hline
  n & m & k & \thead{Scheme} & \thead{Bias} & \thead{MSE} & \thead{Symmetric \\ Credible Interval} & \thead{HPD \\ Credible Interval}  \\ [0.5ex]
 \hline\hline
 20 &   14  &   3   &   I   &   0.109   &   0.269   &   1.888 (94)  &   1.658 (92)\\
    &       &       &   II  &   0.098   &   0.253   &   1.832 (94)  &   1.614 (92)\\
    &       &       &   III &   0.068   &   0.200   &   1.653 (94)  &   1.478 (92)\\
    &       &   5   &   I   &   0.108   &   0.255	&   1.881 (94)  &   1.654 (93)\\
    &       &       &   II  &   0.092   &   0.248   &   1.819 (94)  &   1.604 (92)  \\
    &       &       &   III &   0.065   &   0.192   &   1.644 (94)  &   1.471 (92)\\
    &	18	&	3	&	I	&	0.079	&	0.175	&	1.565 (94)	&	1.414 (93)\\
    &		&		&	II	&	0.085	&	0.190	&	1.574 (94)	&	1.422 (93)\\
    &		&		&	III	&	0.080	&	0.166	&	1.491 (94)	&	1.358 (93)\\
    &     	&   5   &   I   &   0.083   &   0.189   &   1.579 (94)  &   1.425 (93)\\
    &       &       &   II  &   0.072   &   0.176   &   1.541 (94)  &   1.395 (93)\\
    &       &       &   III &   0.066   &   0.155   &   1.467 (94)  &   1.337 (93)\\
 30 &   22  &   6   &   I   &   0.067   &   0.145   &   1.379 (94)  &   1.268 (93)\\
    &       &       &   II  &   0.063   &   0.138   &   1.352 (94)  &   1.245 (93)\\
    &       &       &   III &   0.045   &   0.112   &   1.242 (93)  &   1.155 (92)\\
    &       &   8   &   I   &   0.060   &   0.139   &   1.365 (94)  &   1.255 (93)  \\
    &       &       &   II  &   0.063   &   0.138   &   1.352 (94)  &   1.245 (93)  \\
    &       &       &   III &   0.045   &   0.112   &   1.242 (93)  &   1.155 (92)  \\
    &   26  &   6   &   I   &   0.053   &   0.110   &   1.232 (94)  &   1.147 (94)\\
    &       &       &   II  &   0.051   &   0.104   &   1.216 (94)  &   1.134 (93)\\
    &       &       &   III &   0.042   &   0.087   &   1.139 (95)  &   1.069 (94)\\
    &       &   8   &   I   &   0.053   &   0.110   &   1.232 (94)  &   1.147 (94)\\
    &       &       &   II  &   0.051   &   0.104   &   1.216 (94)  &   1.134 (93)\\
    &       &       &   III &   0.042   &   0.087   &   1.139 (95)  &   1.069 (94)\\
 \hline
 \end{tabular}
 }
\end{table}

\begin{table}[h!]
\captionsetup{font=scriptsize}
\caption {Results of frequentist analysis for $\theta_2$}
\label{freq-theta-2} \centering
\scalebox{0.65}
{
 \begin{tabular}{|c c c c c c c c c|}
 \hline
 n & m & k & \thead{Scheme} & \thead{Bias} & \thead{MSE} & \thead{Exact \\ CI} & NECI	& \thead{Bootstrap \\ CI}   \\ [0.5ex]
 \hline\hline
 20	&	14	&	3	&	I	&	0.172	&	0.561	&	5.952 (96)	&		54	&	4.162 (95)	\\
 	&		&		&	II	&	0.170	&	0.543	&	5.604 (96)	&		45	&	4.029 (95)	\\
 	&		&		&	III	&	0.149	&	0.482	&	4.632 (96)	&		25	&	3.720 (95)	\\
 	&		&	5	&	I	&	0.138	&	0.498	&	5.497 (97)	&		59	&	4.040 (95)	\\
 	&		&		&	II	&	0.175	&	0.540	&	5.773 (96)	&		50	&	4.058 (95)	\\
 	&		&		&	III	&	0.152	&	0.477	&	4.649 (96)	&		27	&	3.724 (95)	\\
 	&	18	&	3	&	I	&	0.163	&	0.484	&	3.796 (96)	&		17	&	3.653 (94)	\\
 	&		&		&	II	&	0.169	&	0.509	&	3.910 (96)	&		21	&	3.575 (94)	\\
 	&		&		&	III	&	0.145	&	0.403	&	3.240 (96)	&		3	&	3.358 (94)	\\
  	&     	&   5   &   I   &   0.151   &   0.452   &   3.713 (96)  &   	23	&	3.595 (95)  \\
    &       &       &   II  &   0.167   &   0.517   &   3.846 (95)  &   	19	&	3.599 (94)  \\
    &       &       &   III &   0.144   &   0.458   &   3.440 (95)  &   	16	&	3.349 (93)   \\
 30 &   22  &   6   &   I   &   0.122   &   0.340   &   2.714 (96)  &   	5	&	2.982 (94)  \\
    &       &       &   II  &   0.122   &   0.346   &   2.687 (95)  &   	4	&	2.877 (94)  \\
    &       &       &   III &   0.098   &   0.276   &   2.361 (95)  &   	3	&	2.487 (94)  \\
    &       &   8   &   I   &   0.119   &   0.346   &   2.744 (96)  &   	3	&	2.986 (95)  \\
    &       &       &   II  &   0.122   &   0.346   &   2.687 (95)  &   	0	&	2.877 (94)  \\
    &       &       &   III &   0.098   &   0.276   &   2.361 (95)  &   	3	&	2.487 (94)  \\
    &   26  &   6   &   I   &   0.108   &   0.296   &   2.311 (95)  &   	0	&	2.551 (94)  \\
    &       &       &   II  &   0.109   &   0.294   &   2.278 (95)  &   	3	&	2.495 (94)  \\
    &       &       &   III &   0.093   &   0.255   &   2.074 (95)  &   	0	&	2.272 (93)  \\
    &       &   8   &   I   &   0.108   &   0.296   &   2.311 (95)  &   	0	&	2.551 (94)  \\
    &       &       &   II  &   0.109   &   0.294   &   2.278 (95)  &   	1	&	2.495 (94)  \\
    &       &       &   III &   0.093   &   0.255   &   2.074 (95)  &   	0	&	2.272 (93)  \\
 \hline
 \end{tabular}
 }
\end{table}
\begin{table}[h!]
\captionsetup{font=scriptsize}
\caption {Results of Bayesian analysis under Prior-I for
$\theta_2$} 
\label{informative-theta-2} \centering
\scalebox{0.7}
{
 \begin{tabular}{|c c c c c c c c|}
 \hline
 n & m & k & \thead{Scheme} & \thead{Bias} & \thead{MSE} & \thead{Symmetric \\ Credible Interval} & \thead{HPD \\ Credible Interval}  \\ [0.5ex]
 \hline\hline
 20	&	14	&	3	&	I	&	0.129	&	0.338	&	2.536 (97)	&	2.216 (94)\\
 	&		&		&	II	&	0.130	&	0.329	&	2.500 (97)	&	2.189 (94)\\
 	&		&		&	III	&	0.118	&	0.310	&	2.362 (96)	&	2.088 (94)\\
 	&		&	5	&	I	&	0.104	&	0.304	&	2.466 (97)	&	2.158 (94)\\
 	&		&		&	II	&	0.134	&	0.326	&	2.511 (97)	&	2.480 (92)\\
 	&		&		&	III	&	0.121	&	0.307	&	2.366 (96)	&	2.092 (94)\\
 	&	18	&	3	&	I	&	0.132	&	0.316	&	2.264 (96)	&	2.020 (94)\\
 	&		&		&	II	&	0.137	&	0.326	&	2.258 (96)	&	2.015 (95)\\
 	&		&		&	III	&	0.123	&	0.280	&	2.140 (96)	&	1.925 (94)\\
  	&     	&   5   &   I   &   0.123   &   0.299   &   2.237 (96)  &   1.997 (94)\\
    &       &       &   II  &   0.135   &   0.334   &   2.261 (96)  &   2.017 (94)\\
    &       &       &   III &   0.119   &   0.307   &   2.143 (95)  &   1.926 (94)\\
 30 &   22  &   6   &   I   &   0.106   &   0.246   &   1.990 (96)  &   1.807 (94)\\
    &       &       &   II  &   0.106   &   0.249   &   1.967 (96)  &   1.789 (94)\\
    &       &       &   III &   0.090   &   0.213   &   1.831 (95)  &   1.680 (94)\\
    &       &   8   &   I   &   0.103   &   0.247   &   1.989 (96)  &   1.805 (95)\\
    &       &       &   II  &   0.106   &   0.249   &   1.967 (96)  &   1.789 (94)\\
    &       &       &   III &   0.090   &   0.213   &   1.831 (95)  &   1.680 (94)\\
    &   26  &   6   &   I   &   0.097   &   0.222   &   1.818 (96)  &   1.671 (95)\\
    &       &       &   II  &   0.099   &   0.223   &   1.805 (96)  &   1.660 (94)\\
    &       &       &   III &   0.087   &   0.202   &   1.701 (95)  &   1.575 (94)\\
    &       &   8   &   I   &   0.097   &   0.222   &   1.818 (96)  &   1.671 (95)\\
    &       &       &   II  &   0.099   &   0.223   &   1.805 (96)  &   1.660 (94)\\
    &       &       &   III &   0.087   &   0.202   &   1.701 (95)  &   1.575 (94)\\
 \hline
 \end{tabular}
 }
\end{table}

\begin{table}[h!]
\captionsetup{font=scriptsize}
\caption {Results of Bayesian analysis under Prior-II  for
$\theta_2$} 
\label{non-informative-theta-2} \centering
\scalebox{0.7}
{
 \begin{tabular}{|c c c c c c c c|}
 \hline
 n & m & k & \thead{Scheme} & \thead{Bias} & \thead{MSE} & \thead{Symmetric \\ Credible Interval} & \thead{HPD \\ Credible Interval}  \\ [0.5ex]
 \hline\hline
 20	&	14	&	3	&	I	&	0.172	&	0.561	&	2.960 (95)	&	2.509 (92)\\
 	&		&		&	II	&	0.170	&	0.543	&	2.902 (95)	&	2.470 (92)\\
 	&		&		&	III	&	0.149	&	0.482	&	2.686 (94)	&	2.318 (92)\\
 	&		&	5	&	I	&	0.138	&	0.498	&	2.852 (95)	&	2.424 (91)\\
 	&		&		&	II	&	0.175	&	0.540	&	2.915 (95)	&	2.480 (92)\\
 	&		&		&	III	&	0.152	&	0.477	&	2.691 (94)	&	2.323 (92)\\
 	&	18	&	3	&	I	&	0.163	&	0.484	&	2.549 (94)	&	2.226 (93)\\
 	&		&		&	II	&	0.169	&	0.509	&	2.546 (94)	&	2.226 (93)\\
 	&		&		&	III	&	0.145	&	0.403	&	2.367 (94)	&	2.093 (93)\\
  	&     	&   5   &   I   &   0.151   &   0.452   &   2.508 (94)  &   2.194 (92)\\
    &       &       &   II  &   0.167   &   0.517   &   2.554 (94)  &   2.230 (92)\\
    &       &       &   III &   0.144   &   0.458   &   2.387 (94)  &   2.105 (92)\\
 30 &   22  &   6   &   I   &   0.122   &   0.340   &   2.167 (95)  &   1.941 (93)\\
    &       &       &   II  &   0.122   &   0.346   &   2.141 (94)  &   1.920 (93)\\
    &       &       &   III &   0.098   &   0.276   &   1.960 (94)  &   1.780 (92)\\
    &       &   8   &   I   &   0.119   &   0.346   &   2.170 (95)  &   1.942 (93)\\
    &       &       &   II  &   0.122   &   0.346   &   2.141 (94)  &   1.920 (93)\\
    &       &       &   III &   0.098   &   0.276   &   1.960 (94)  &   1.780 (92)\\
    &   26  &   6   &   I   &   0.108   &   0.296   &   1.949 (95)  &   1.773 (94)\\
    &       &       &   II  &   0.109   &   0.294   &   1.932 (94)  &   1.760 (93)\\
    &       &       &   III &   0.093   &   0.255   &   1.801 (94)  &   1.655 (93)\\
    &       &   8   &   I   &   0.108   &   0.296   &   1.949 (95)  &   1.773 (94)\\
    &       &       &   II  &   0.109   &   0.294   &   1.932 (94)  &   1.760 (93)\\
    &       &       &   III &   0.093   &   0.255   &   1.801 (94)  &   1.655 (93)\\
 \hline
 \end{tabular}
 }
\end{table}

\subsection{\sc Data analysis}

In this section we consider a real data set mainly for illustrative purposes. The data is taken from an experiment conducted by 
Dr. H.E. Walburg, Jr., of the Oak Ridge National Laboratory (see Hoel \cite{hoel}).  A group of male mice received a radiation 
dose of 300r at age 5-6 weeks.  The causes of deaths were classified as (1) Thymic Lymphoa, (2) Reticulum Cell Sarcoma and 
(3) Other causes.  For our analysis we have consider Reticulum Cell Sarcoma as Cause-1 and combined all other causes as Cause-2.  
We have generated a generalized progressively censored sample based on the following scheme: $n=77, k=20, m=25, T=700, 
R_1=\ldots=R_{24}=2, R_{25}=4$.  The data set is presented below.  Here the first component represents the failure time and the second
component represents the cause of failure.

(40, 2), (42, 2), (62, 2), (163, 2), (179, 2), (206, 2), (222, 2), (228, 2), (252, 2), (259, 2), (318, 1), (385, 2), (407, 2), (420, 2), (462, 2), (507, 2), (517, 2), (524, 2), (525, 1), (528, 1), (536, 1), (605, 1), (612, 1), (620, 2), (621, 1). 

From the data set we obtain $\displaystyle W=\sum_{i=1}^m (R_i+1)z_i=28962.$  $D_1=7, D_2=18$.  Since we do not have any prior 
information about the unknown parameters, we use Prior-II in this case.  The Bayes estimates of $\theta_1$ and $\theta_2$ are 
$\widehat{\theta}_{1B}=4137.429$ and $\widehat{\theta}_{2B}=1609$, respectively.  These are also the MLEs of $\theta_1$ and 
$\theta_2$.   The Bayes estimates of $\lambda_{1B}$ and $\lambda_{2B}$ are $\widehat{\lambda}_1=0.0003$ and $\widehat{\lambda}_2 = 
0.0006$, respectively.  We present different 95\% confidence and credible intervals of $\theta_1$ and $\theta_2$ in Table 
\ref{tab:title}.  In Figure \ref{cset} we present the 95\% credible set of $\theta_1$ and $\theta_2$.
\begin{table}[h!]
\caption {Results for real data} \label{tab:title}
\centering
\scalebox{0.8}{
\begin{tabular}{|c| c| c| c|} 
 \hline     
  & Exact CI &  Bootstrap CI &  HPD CRI\\
 \hline
 $\theta_1$ & 2017.686, 10397.358 &  2061.129, 12220.355 & 1715.194, 7480.241\\
$\theta_2$ & 1018.497, 2790.006 &  976.663, 2749.781 & 939.656, 2363.621\\ 
 \hline
 \end{tabular}}
 \end{table} 
\begin{figure}[h]
\includegraphics[scale=0.5]{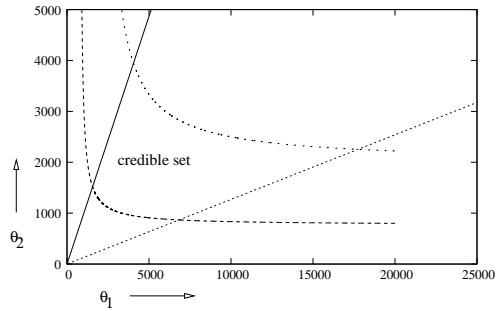}
\centering
\captionsetup{font=scriptsize}
\caption{Credible set of $\theta_1$ and $\theta_2$   \label{cset}}
\end{figure}

\section{\sc Conclusion}

In this paper we consider the analysis of generalized progressive hybrid censored data in presence of competing risks.  It is assumed
that the competing causes of failures satisfy the latent failure time model assumption of Cox \cite{cox}.  We further assume that the
lifetime distributions of the competing causes are exponentially distributed with different scale parameters.  We carry out both 
the frequentist and Bayesian analysis of the unknown parameters, and it is observed that the performances of the Bayes estimates
are better than the MLEs.  It should be mention that recently Gorny and Cramer \cite{GC:2016} developed the exact likelihood inference
for exponential distributions under generalized hybrid censoring schemes using splines.  The treatments are completely different than 
ours.  It will be interesting to develop the exact inference results in presence of competing using splines approach.  More work is 
needed in that direction.  One natural extension of the present work is when the lifetime distributions of the competing causes
follow Weibull distributions.  It is a more challenging problem.  Work is in progress and it will reported later.

\section*{\sc Acknowledgements:} The would like to thank the referees for their constructive suggestions which have helped us
to improve the manuscript significantly.

 \section*{Appendix: The Proof of the Main Theorem}

First we derive the distribution function of $\widehat{\theta}_1$ which is given 
below. 
\begin{eqnarray}
F_{\widehat{\theta}_1| D_1>0}(x)& = &P(\widehat{\theta}_1 \leq x | D_1>0)   \nonumber  \\
& = & P(\widehat{\theta}_1 \leq x, A | D_1>0)+P(\widehat{\theta}_1 \leq x, B | D_1>0) +P(\widehat{\theta}_1 \leq x, C | D_1>0) \nonumber \\
& = &\sum_{j=k}^{m-1} \sum_{i=1}^j P(\widehat{\theta}_1 \leq x| J=j, D_1=i)P(J=j, D_1=i | D_1>0)  \nonumber \\
& &+\sum_{i=1}^k P(\widehat{\theta}_1 \leq x| B, D_1=i) P(B, D_1=i | D_1>0)  \nonumber  \\
& &+ \sum_{i=1}^m P(\widehat{\theta}_1 \leq x| C, D_1=i) P(C, D_1=i | D_1>0),    \label{dist}
\end{eqnarray}
where, 
\[
A=\{Z_{k:m:n}<T<Z_{m:m:n}\},\ B=\{T<Z_{k:m:n}<Z_{m:m:n}\}, \ C=\{Z_{k:m:n}<Z_{m:m:n}<T\}.
\]
Now to compute the terms on the right hand side of (\ref{dist}), we need the following Lemmas. 
\enlargethispage{1 cm}
\noindent {\sc Lemma: 1} ~The joint distribution of $Z_{1:m:n},\ldots, Z_{J:m:n}$ given $J=j, D_1=i$ for $i=1,\ldots,j$ and $j=k,\dots, m-1$ at $z_1,\ldots, z_j$, is given by,
\begin{align}
&f_{Z_{1:m:n},\ldots, Z_{j:m:n}|J=j, D_1=i}(z_1, \ldots, z_j)\nonumber\\
&=\frac{\prod_{v=1}^j \gamma_v}{P(J=j, D_1=i)} \Big(\frac{1}{\theta_1}\Big)^i \Big(\frac{1}{\theta_2}\Big)^{j-i} e^{-\frac{1}{\theta} (\sum_{s=1}^j z_s (1+R_s) + TR^*_j)}.\label{lemma1}
\end{align}

\noindent {\sc Proof of {\sc Lemma 1:}}
For $j=k,k+1,\ldots m \ \ \hbox{and} \ \ i=1,2,\ldots, j $, consider left side of (\ref{lemma1}),
\begin{align}
&P(z_1<Z_{1:m:n}<z_1+dz_1,\ldots,z_j<Z_{j:m:n}<z_j+dz_j|J=j, D_1=i) \nonumber \\
&=\frac{P(z_1<Z_{1:m:n}<z_1+dz_1,\ldots,z_j<Z_{j:m:n}<z_j+dz_j,J=j, D_1=i)}{P(J=j, D_1=i)} \label{lemma1_1}    \\
\nonumber
\end{align}
\noindent Note that, the event $\{z_1<Z_{1:m:n}<z_1+dz_1,\ldots,z_j<Z_{j:m:n}<z_j+dz_j,J=j, D_1=i)\}$ for $i=1,\ldots,j; j=k,k+1,\ldots,m$ is nothing but the failure times of $j$ units till time point $T$ and out of them $i$ units have failed due to Cause-1. The probability of this event is the likelihood contribution of the data when $T^*=T$. Thus (\ref{lemma1_1}) becomes, 
\begin{align*}
\prod_{v=1}^j \gamma_v \bigg(\frac{1}{\theta_1}\bigg)^i \bigg(\frac{1}{\theta_2}\bigg)^{j-i}\frac{e^{-\frac{1}{\theta}\big[\sum_{s=1}^{j} (1+R_s)z_s + TR^*_j\big]}}{P(J=j,D_1=i)} \; dz_1\ldots dz_j.
\end{align*}
\qed

\noindent {\sc Lemma 2:} ~The joint distribution of $Z_{1:m:n},\ldots, Z_{k:m:n}$ given $T<Z_{k:m:n}<Z_{m:m:n}, D_1=i$ for $i=1,\ldots,k$ at $z_1,\ldots, z_k$, is given by
\begin{align}
&f_{Z_{1:m:n},\ldots, Z_{k:m:n}|T<Z_{k:m:n}<Z_{m:m:n}, D_1=i}(z_1, \ldots, z_k) \nonumber \\
&=\frac{\prod_{v=1}^k \gamma_v}{P(T<Z_{k:m:n}<Z_{m:m:n}, D_1=i)} \Big(\frac{1}{\theta_1}\Big)^i \Big(\frac{1}{\theta_2}\Big)^{k-i} e^{-\frac{1}{\theta} (\sum_{s=1}^{k-1} z_s (1+R_s) + z_k (1+R^*_k))}. \label{lemma2}
\end{align}

\noindent {\sc Proof of Lemma 2:}
For $i=1,2, \ldots, k$, \noindent consider left side of (\ref{lemma2}),
\begin{align}
&P(z_1<Z_{1:m:n}<z_1+dz_1,\ldots,z_k<Z_{k:m:n}<z_k+dz_k|T<Z_{k:m:n}<Z_{m:m:n},D_1=i)\nonumber \\
&=\frac{P(z_1<Z_{1:m:n}<z_1+dz_1,\ldots,z_k<Z_{k:m:n}<z_k+dz_k,T<Z_{k:m:n}<Z_{m:m:n},D_1=i)}{P(T<Z_{k:m:n}<Z_{m:m:n},D_1=i)} \label{lemma2_2} \\
\nonumber
\end{align}
\noindent Note that, the event $\{z_1<Z_{1:m:n}<z_1+dz_1,\ldots,z_k<Z_{k:m:n}<z_k+dz_k,T<Z_{k:m:n}, D_1=i)\}$ for $i=1,\ldots,k$ is nothing but the failure times of $k$ units till the experiment termination point $Z_{k:m:n}$ and out of them $i$ units have failed due to Cause-1. The probability of this event is the likelihood contribution of the data when $T^*=Z_{k:m:n}$. Thus (\ref{lemma2_2}) becomes, 
\begin{align*}
\prod_{v=1}^k \gamma_v\bigg(\frac{1}{\theta_1}\bigg)^i \bigg(\frac{1}{\theta_2}\bigg)^{k-i}\frac{e^{-\frac{1}{\theta}\big[\sum_{s=1}^{k-1} (1+R_s)z_s + z_k(1+R^*_k)\big]}}{P(T<Z_{k:m:n}<Z_{m:m:n},D_1=i)} \; dz_1\ldots dz_k.
\end{align*}
\qed

\pagebreak
\noindent {\sc Lemma 3:} ~The joint distribution of $Z_{1:m:n},\ldots, Z_{m:m:n}$ given $Z_{k:m:n}<Z_{m:m:n}<T, D_1=i$ for $i=1,\ldots,m$ at $z_1,\ldots, z_m$, is given by,
\begin{align}
&f_{Z_{1:m:n},\ldots, Z_{m:m:n}|Z_{k:m:n}<Z_{m:m:n}<T, D_1=i}(z_1, \ldots, z_m) \nonumber \\
&=\frac{\prod_{v=1}^m \gamma_v}{P(Z_{k:m:n}<Z_{m:m:n}<T, D_1=i)} \Big(\frac{1}{\theta_1}\Big)^i \Big(\frac{1}{\theta_2}\Big)^{m-i} e^{-\frac{1}{\theta}\sum_{s=1}^{m} (1+R_s)z_s}.\label{lemma3}
\end{align}

\noindent {\sc Proof of Lemma 3:}
For $i=1,2,\ldots m$, consider left side of (\ref{lemma3}),
\begin{align}
&P(z_1<Z_{1:m:n}<z_1+dz_1,\ldots,z_m<Z_{m:m:n}<z_m+dz_m|Z_{k:m:n}<Z_{m:m:n}<T,D_1=i) \nonumber \\
&=\frac{P(z_1<Z_{1:m:n}<z_1+dz_1,\ldots,z_m<Z_{m:m:n}<z_m+dz_m,Z_{k:m:n}<Z_{m:m:n}<T,D_1=i)}{P(Z_{k:m:n}<Z_{m:m:n}<T,D_1=i)} \label{lemma3_3}\\
\nonumber
\end{align}
\noindent Note that, the event $\{z_1<Z_{1:m:n}<z_1+dz_1,\ldots,z_m<Z_{m:m:n}<z_m+dz_m,Z_{k:m:n}<Z_{m:m:n}<T, D_1=i)\}$ for $i=1,\ldots,m$ is nothing but the failure times of $m$ units till the experiment termination point $Z_{m:m:n}$ and out of them $i$ units have failed due to Cause-1. The probability of this event is the likelihood contribution of the data when $T^*=Z_{m:m:n}$. Thus \ref{lemma3_3} becomes,
\begin{align}
\prod_{v=1}^m \gamma_v\bigg(\frac{1}{\theta_1}\bigg)^i \bigg(\frac{1}{\theta_2}\bigg)^{m-i}\frac{e^{-\frac{1}{\theta}\sum_{s=1}^{m} (1+R_s)z_s}}{P(Z_{k:m:n}<Z_{m:m:n}<T,D_1=i)} \; dz_1\ldots dz_m.
\end{align}
\qed

\begin{theorem}
The conditional moment generating function of $\widehat{\theta}_1$ given $J=j, D_1=i$ for $i=1,\ldots, j$ and 
$j=k,\ldots, m-1$ is given by
\begin{align*}
&\quad E(e^{t\widehat{\theta}_1}|J=j, D_1=i)\\
&= \prod_{v=1}^j \frac{\gamma_v}{P(J=j,D_1=i)}  \bigg(\frac{1}{\theta_1}\bigg)^i\bigg(\frac{1}{\theta_2}\bigg)^{j-i} \Big(\frac{1}{\theta}-\frac{t}{i}\Big)^{-j} \times\\
& \quad \sum_{v=0}^j \frac{(-1)^v e^{-T(\frac{1}{\theta}-\frac{t}{i})\gamma_{j-v+1}}}{\{\prod_{h=1}^v (\gamma_{j+1-v}-\gamma_{j+1-v+h})\}\{\prod_{h=1}^{j-v}  (\gamma_h-\gamma_{j-v+1})\}}.
\end{align*}
\end{theorem}
\begin{proof}
\begin{align*}
&\quad E[e^{t\widehat{\theta}_1}|J=j,D_1=i]\\
&\quad =\frac{\prod_{v=1}^j \gamma_v}{P(J=j,D_1=i)} \bigg(\frac{1}{\theta_1}\bigg)^i\bigg(\frac{1}{\theta_2}\bigg)^{j-i}\times\\
&\quad \int_{0}^T \int_{0}^{z_j}\ldots \int_{0}^{z_2} e^{-z_1(1+R_1)(\frac{1}{\theta}-\frac{t}{i})}   \ldots e^{-z_j(1+R_j)(\frac{1}{\theta}-\frac{t}{i})} e^{-TR^*_j(\frac{1}{\theta} -\frac{t}{i})} \ dz_1\ldots dz_j \\
&~~\text{The above equality follows using Lemma 1,}\\
&\quad =\frac{\prod_{v=1}^j \gamma_v}{P(J=j,D_1=i)} \bigg(\frac{1}{\theta_1}\bigg)^i\bigg(\frac{1}{\theta_2}\bigg)^{j-i} e^{-TR^*_j(\frac{1}{\theta}-\frac{t}{i})}\times\\
& \quad \int_{0}^T \int_{0}^{z_j}\ldots \int_{0}^{z_2} e^{-z_1} e^{-z_1(1+R_1)(\frac{1}{\theta}-\frac{t}{i})-1}   \ldots e^{-z_j} e^{-z_{j}(1+R_{j})(\frac{1}{\theta}-\frac{t}{i})-1}dz_1\ldots dz_j \\
& =\frac{\prod_{v=1}^j \gamma_v}{P(J=j,D_1=i)} \bigg(\frac{1}{\theta_1}\bigg)^i\bigg(\frac{1}{\theta_2}\bigg)^{j-i} \Big(\frac{1}{\theta}-\frac{t}{i}\Big)^{-j} \times \\
&\quad \sum_{v=0}^j \frac{(-1)^v e^{-T(\frac{1}{\theta}-\frac{t}{i})\gamma_{j-v+1}}}{\{\prod_{h=1}^v (\gamma_{j+1-v}-\gamma_{j+1-v+h})\}\{\prod_{h=1}^{j-v}  (\gamma_h-\gamma_{j-v+1})\}} \\
&~~\text{The last equality follows using Lemma 1 of Balakrishnan et al.\cite{chandrasekar}}.
\end{align*}
\end{proof}

\begin{corollary}
The conditional distribution of $\widehat{\theta}_1$ given $J=j, D_1=i$ for $i=1,\ldots,j$ and $j=k,\ldots,m-1$ is given by,
\begin{align*}
f_{\widehat{\theta}_1|J=j,D_1=i}(x)&=\prod_{v=1}^j \frac{\gamma_v}{P(J=j,D_1=i)}  \bigg(\frac{1}{\theta_1}\bigg)^i\bigg(\frac{1}{\theta_2}\bigg)^{j-i} \Big(\frac{1}{\theta}\Big)^{-j}\times \\
& \quad \sum_{v=0}^j \frac{(-1)^v e^{-\frac{T}{\theta}\gamma_{j-v+1}}}{\{\prod_{h=1}^v (\gamma_{j+1-v}-\gamma_{j+1-v+h})\}\{\prod_{h=1}^{j-v}  (\gamma_h-\gamma_{j-v+1})\}}\times\\
& \quad f_G\Big(x;\frac{T}{i}\gamma_{j-v+1},j,\frac{i}{\theta}\Big).
\end{align*}
\end{corollary}
\enlargethispage{1 cm}
\begin{theorem}
The conditional moment generating function of $\widehat{\theta}_1$ given $T<Z_{k:m:n}<Z_{m:m:n},\\ D_1=i$ for $i=1,\ldots, k$ is 
given by,
\begin{align*}
&\quad E(e^{t\widehat{\theta}_1}|T<Z_{k:m:n}<Z_{m:m:n}, D_1=i)\\
&\quad=\frac{\prod_{v=1}^k \gamma_v}{P(T<Z_{k:m:n}<Z_{m:m:n},D_1=i)} \bigg(\frac{1}{\theta_1}\bigg)^i\bigg(\frac{1}{\theta_2}\bigg)^{k-i} \Big(\frac{1}{\theta}-\frac{t}{i}\Big)^{-k} \times\\
&\quad \sum_{v=0}^{k-1} \frac{(-1)^v e^{-T(\frac{1}{\theta}-\frac{t}{i})\gamma_{k-v}}}{\{\prod_{j=1}^v  (\gamma_{k-v}-\gamma_{k-v+j})\}\{\prod_{j=1}^{k-1-v}  (\gamma_j-\gamma_{k-v})\}\gamma_{k-v}}.
\end{align*}
\end{theorem}
\begin{proof}
\begin{align*}
&\quad E[e^{t\widehat{\theta}_1}|T<Z_{k:m:n}<Z_{m:m:n},D_1=i]\\
&=\frac{\prod_{v=1}^k \gamma_v}{P(T<Z_{k:m:n}<Z_{m:m:n},D_1=i)} \bigg(\frac{1}{\theta_1}\bigg)^i\bigg(\frac{1}{\theta_2}\bigg)^{k-i}\times\\
&\quad \int_{T}^\infty \int_{0}^{z_k}\ldots \int_{0}^{z_2} e^{-z_1(1+R_1)(\frac{1}{\theta}-\frac{t}{i})}   \ldots e^{-z_{k-1}(1+R_{k-1})(\frac{1}{\theta}-\frac{t}{i})}e^{-z_k(1+R^*_k)(\frac{1}{\theta}-\frac{t}{i})} dz_1\ldots dz_k  \\
&~~\text{The above equality follows using Lemma 2,}\\
&=\frac{\prod_{v=1}^k \gamma_v}{P(T<Z_{k:m:n}<Z_{m:m:n},D_1=i)} \bigg(\frac{1}{\theta_1}\bigg)^i\bigg(\frac{1}{\theta_2}\bigg)^{k-i}\times\\
&\quad \int_{T}^\infty \Big[\int_{0}^{z_k}\ldots \int_{0}^{z_2} e^{-z_1} e^{-z_1(1+R_1)(\frac{1}{\theta}-\frac{t}{i})-1} \ldots e^{-z_{k-1}} e^{-z_{k-1}(1+R_{k-1})(\frac{1}{\theta}-\frac{t}{i})-1}dz_1\ldots dz_{k-1} \Big]\\
{} &\ \ \ \ \ \ \ e^{-z_k(1+R^*_k)(\frac{1}{\theta}-\frac{t}{i})} dz_k\\
&=\frac{\prod_{v=1}^k \gamma_v}{P(T<Z_{k:m:n}<Z_{m:m:n},D_1=i)} \bigg(\frac{1}{\theta_1}\bigg)^i\bigg(\frac{1}{\theta_2}\bigg)^{k-i} \frac{1}{(\frac{1}{\theta}-\frac{t}{i})^{k-1}}\times\\
&\quad \int_T^\infty \sum_{v=0}^{k-1} \frac{(-1)^v e^{-z_k(\frac{1}{\theta}-\frac{t}{i})\gamma_v}}{\{\prod_{j=1}^v  (\gamma_{k-v}-\gamma_{k-v+j})\}\{\prod_{j=1}^{k-1-v}  (\gamma_j-\gamma_{k-v})\}} dz_k\\
&~~\text{The last equality follows using Lemma 1 of Balakrishnan et al.\cite{chandrasekar}}\\
&=\frac{\prod_{v=1}^k \gamma_v}{P(T<Z_{k:m:n}<Z_{m:m:n},D_1=i)} \bigg(\frac{1}{\theta_1}\bigg)^i\bigg(\frac{1}{\theta_2}\bigg)^{k-i} \Big(\frac{1}{\theta}-\frac{t}{i}\Big)^{-k} \times \\
&\quad \sum_{v=0}^{k-1} \frac{(-1)^v e^{-T(\frac{1}{\theta}-\frac{t}{i})\gamma_{k-v}}}{\{\prod_{j=1}^v  (\gamma_{k-v}-\gamma_{k-v+j})\}\{\prod_{j=1}^{k-1-v}  (\gamma_j-\gamma_{k-v})\}\gamma_{k-v}}.
\end{align*}
\end{proof}
\begin{corollary}
The conditional distribution of $\widehat{\theta}_1$ given $T<Z_{k:m:n}<Z_{m:m:n}, D_1=i$ for $i=1,\ldots, k$ is given by,
\begin{align*}
f_{\widehat{\theta}_1|T<Z_{k:m:n}<Z_{m:m:n},D_1=i}(x)
&\quad=\frac{\prod_{v=1}^k \gamma_v}{P(T<Z_{k:m:n}<Z_{m:m:n},D_1=i)} \bigg(\frac{1}{\theta_1}\bigg)^i\bigg(\frac{1}{\theta_2}\bigg)^{k-i} \Big(\frac{1}{\theta}\Big)^{-k}\times \\
&\quad \sum_{v=0}^{k-1} \frac{(-1)^v e^{-\frac{T}{\theta}\gamma_{k-v}}}{\{\prod_{j=1}^v  (\gamma_{k-v}-\gamma_{k-v+j})\}\{\prod_{j=1}^{k-1-v}  (\gamma_j-\gamma_{k-v})\}\gamma_{k-v}}\times\\
&\quad f_G\Big(x;\frac{T}{i}\gamma_{k-v},k,\frac{i}{\theta}\Big).
\end{align*}
\end{corollary}
\pagebreak
\begin{theorem}
The moment generating function of $\widehat{\theta}_1$ given $Z_{k:m:n}<Z_{m:m:n}<T, D_1=i$ for $i=1,\ldots,m$ is given by,
\begin{align*}
&\quad E(e^{t\widehat{\theta}_1}|Z_{k:m:n}<Z_{m:m:n}<T, D_1=i)\\
&=\frac{\prod_{v=1}^m \gamma_v}{P(Z_{k:m:n}<Z_{m:m:n}<T,D_1=i)} \bigg(\frac{1}{\theta_1}\bigg)^i\bigg(\frac{1}{\theta_2}\bigg)^{m-i} \Big(\frac{1}{\theta}-\frac{t}{i}\Big)^{-m}\times \\
&\quad\sum_{v=0}^{m} \frac{(-1)^v e^{-T(\frac{1}{\theta}-\frac{t}{i})(\gamma_{m-v+1}-\gamma_{m+1})}}{\{\prod_{j=1}^v  (\gamma_{m-v+1}-\gamma_{m-v+j+1})\}\{\prod_{j=1}^{m-v}  (\gamma_j-\gamma_{m-v+1})\}}.
\end{align*}
\end{theorem}
\begin{proof}
\begin{align*}
& \quad E[e^{t\widehat{\theta}_1}|Z_{k:m:n}<Z_{m:m:n}<T,D_1=i]\\
&  =\frac{\prod_{v=1}^m \gamma_v}{P(Z_{k:m:n}<Z_{m:m:n}<T,D_1=i)} \bigg(\frac{1}{\theta_1}\bigg)^i\bigg(\frac{1}{\theta_2}\bigg)^{m-i}\times\\
& \quad \int_{0}^T \int_{0}^{z_m}\ldots \int_{0}^{z_2} e^{-z_1(1+R_1)(\frac{1}{\theta}-\frac{t}{i})} \ldots e^{-z_{m-1}(1+R_{m-1})(\frac{1}{\theta}-\frac{t}{i})}e^{-z_m(1+R_m)(\frac{1}{\theta}-\frac{t}{i})} dz_1\ldots dz_m  \\
&~~\text{The above equality follows using Lemma 3,}\\
&  =\frac{\prod_{v=1}^m \gamma_v}{P(Z_{k:m:n}<Z_{m:m:n}<T,D_1=i)} \bigg(\frac{1}{\theta_1}\bigg)^i\bigg(\frac{1}{\theta_2}\bigg)^{m-i}\times\\
& \quad \int_{0}^T \ldots \int_{0}^{z_2} e^{-z_1} e^{-z_1(1+R_1)(\frac{1}{\theta}-\frac{t}{i})-1} \ldots e^{-z_m} e^{-z_m(1+R_m)(\frac{1}{\theta}-\frac{t}{i})-1} dz_1\ldots dz_m\\
& =\frac{\prod_{v=1}^m \gamma_v}{P(Z_{k:m:n}<Z_{m:m:n}<T,D_1=i)} \bigg(\frac{1}{\theta_1}\bigg)^i\bigg(\frac{1}{\theta_2}\bigg)^{m-i} \Big(\frac{1}{\theta}-\frac{t}{i}\Big)^{-m}\times\\
& \quad \sum_{v=0}^{m} \frac{(-1)^v e^{-T(\frac{1}{\theta}-\frac{t}{i})(\gamma_{m-v+1}-\gamma_{m+1})}}{\{\prod_{j=1}^v  (\gamma_{m-v+1}-\gamma_{m-v+j+1})\}\{\prod_{j=1}^{m-v}  (\gamma_j-\gamma_{m-v+1})\}} \\
&~~\text{The last equality follows using Lemma 1 of Balakrishnan et al.\cite{chandrasekar}}.
\end{align*}
\end{proof}

\begin{corollary}
The conditional distribution of $\widehat{\theta}_1$ given $Z_{k:m:n}<Z_{m:m:n}<T, D_1=i$ for $i=1,\ldots, m$ is given by,
\begin{align*}
f_{\widehat{\theta}_1|Z_{k:m:n}<Z_{m:m:n}<T,D_1=i}(x)&\quad=\frac{\prod_{v=1}^m \gamma_v}{P(Z_{k:m:n}<Z_{m:m:n}<T,D_1=i)} \bigg(\frac{1}{\theta_1}\bigg)^i\bigg(\frac{1}{\theta_2}\bigg)^{m-i} \Big(\frac{1}{\theta}\Big)^{-m}\times\\
&\quad\sum_{v=0}^{m} \frac{(-1)^v e^{-\frac{T}{\theta}(\gamma_{m-v+1}-\gamma_{m+1})}}{\{\prod_{j=1}^v  (\gamma_{m-v+1}-\gamma_{m-v+j+1})\}\{\prod_{j=1}^{m-v}  (\gamma_j-\gamma_{m-v+1})\}}\times\\
&\quad f_G\Big(x;\frac{T}{i}(\gamma_{m-v+1}-\gamma_{m+1}),m,\frac{i}{\theta}\Big).
\end{align*}
\end{corollary}

\noindent {\sc Proof of Theorem 1:} Combining corollaries 1 - 3, we get the first part of Theorem 1.   \qed

\noindent \underline{Derivation of $P(D_1=0)$.}
\begin{align*}
P(D_1=0)&=P(D_1=0,Z_{k:m:n}<T<Z_{m:m:n})+P(D_1=0,T<Z_{k:m:n}<Z_{m:m:n})\\
{}&+P(D_1=0,Z_{k:m:n}<Z_{m:m:n}<T)\\
{}&=P(Z_{k:m:n}<T<Z_{m:m:n})P(D_1=0|Z_{k:m:n}<T<Z_{m:m:n})\\
{}&+P(T<Z_{k:m:n}<Z_{m:m:n})P(D_1=0|T<Z_{k:m:n}<Z_{m:m:n})\\
{}&+P(Z_{k:m:n}<Z_{m:m:n}<T)P(D_1=0|Z_{k:m:n}<Z_{m:m:n}<T).
\end{align*}
\noindent We find each of the above probabilities separately.
\begin{align*}
&P(Z_{k:m:n}<T<Z_{m:m:n})\\
&=\sum_{j=k}^{m-1} P(Z_{j:m:n}<T<Z_{j+1:m:n})\\
&=\sum_{j=k}^{m-1} \prod_{v=1}^{j+1} \gamma_v \Big(\frac{1}{\theta}\Big)^{j+1}\int_{T}^{\infty}\int_{0}^T\ldots \int_{0}^{z_2}e^{-\frac{1}{\theta}\sum_{i=1}^{j} z_i(1+R_i)} e^{-\frac{1}{\theta} z_{j+1}(1+R^*_{j+1})} dz_1 \ldots dz_j dz_{j+1}\\
& =\sum_{j=k}^{m-1} \prod_{v=1}^{j+1} \gamma_v  \sum_{u=0}^{j} \frac{(-1)^u e^{-\frac{T}{\theta}\gamma_{j-u+1}}}{\{\prod_{v=1}^u  (\gamma_{j-u+1}-\gamma_{j+v-u+1})\}\{\prod_{v=1}^{j-u}  (\gamma_v-\gamma_{j-u+1})\}\{\gamma_{j-u+1}-u\}}.\\
\end{align*}
\begin{align*}
P(D_1=0|J=j)=\Big(\frac{\theta_1}{\theta_1+\theta_2}\Big)^j.
\end{align*}
\begin{align*}
&P(T<Z_{k:m:n}<Z_{m:m:n})\\
&=\prod_{v=1}^k \gamma_v \Big(\frac{1}{\theta}\Big)^k \int_T^{\infty} \int_{0}^{z_k}\ldots \int_0^{z_2} e^{-\frac{1}{\theta}\Big(\sum_{i=1}^k z_i(1+R_i)+z_k\gamma_k\Big)} dz_1\ldots dz_{k-1} dz_k\\
&=\prod_{v=1}^k \gamma_v \sum_{v=0}^{k-1} \frac{(-1)^v e^{-\frac{T}{\theta}\gamma_{k-v}}}{\{\prod_{j=1}^v  (\gamma_{k-v}-\gamma_{k-v+j})\}\{\prod_{j=1}^{k-1-v}  (\gamma_j-\gamma_{k-v})\}\gamma_{k-v}}.\\
\end{align*}
\begin{align*}
P(D_1=0|T<Z_{k:m:n}<Z_{m:m:n}) =\Big(\frac{\theta_1}{\theta_1+\theta_2}\Big)^k
\end{align*}
\begin{align*}
&P(Z_{k:m:n}<Z_{m:m:n}<T)\\
&=\prod_{v=1}^m \gamma_v \Big(\frac{1}{\theta}\Big)^m \int_0^{T} \int_{0}^{z_m}\ldots \int_0^{z_2} e^{-\frac{1}{\theta}\sum_{i=1}^m z_i(1+R_i)}dz_1\ldots dz_{m-1} dz_m\\
&=\prod_{v=1}^m \gamma_v \sum_{v=0}^{m} \frac{(-1)^v e^{-\frac{T}{\theta}(\gamma_{m-v+1}-\gamma_{m+1})}}{\{\prod_{j=1}^v  (\gamma_{m-v+1}-\gamma_{m-v+j+1})\}\{\prod_{j=1}^{m-v}  (\gamma_j-\gamma_{m-v+1})\}}.\\
\end{align*}
\begin{align*}
P(D_1=0|Z_{k:m:n}<Z_{m:m:n}<T)=\Big(\frac{\theta_1}{\theta_1+\theta_2}\Big)^m.
\end{align*}

\end{document}